\documentclass[a4paper,twocolumn,pra]{revtex4-1}
\usepackage{amsmath}
\usepackage{amssymb}
\usepackage{graphicx}
\usepackage{amsthm}
\usepackage{hyperref}
\usepackage{bbold}

\begin{document}

\newcommand*{\cl}[1]{{\mathcal{#1}}}
\newcommand*{\bb}[1]{{\mathbb{#1}}}
\newcommand{\ket}[1]{|#1\rangle}
\newcommand{\bra}[1]{\langle#1|}
\newcommand{\inn}[2]{\langle#1|#2\rangle}
\newcommand{\proj}[2]{| #1 \rangle\!\langle #2 |}
\newcommand*{\tn}[1]{{\textnormal{#1}}}
\newcommand*{\1}{{\mathbb{1}}}
\newcommand{\T}{\mbox{$\textnormal{Tr}$}}
\newcommand{\todo}[1]{\textcolor[rgb]{0.99,0.1,0.3}{#1}}

\theoremstyle{plain}
\newtheorem{prop}{Proposition}
\newtheorem{proposition}{Proposition}
\newtheorem{theorem}{Theorem}
\newtheorem{lemma}{Lemma}
\newtheorem{remark}{Remark}

\theoremstyle{definition}
\newtheorem{definition}{Definition}

\title{Weighted $p$-R\'{e}nyi Entropy Power Inequality: Information Theory to Quantum Shannon Theory}
\author{Junseo Lee}
\email{homology.manifold@gmail.com}
\affiliation{School of Electrical and Electronic Engineering, Yonsei University, Seoul 03722, Korea}
\affiliation{Quantum Computing R\&D, Norma Inc., Seoul 04799, Korea}
\author{Hyeonjun Yeo}
\email{duguswns11@snu.ac.kr}
\affiliation{Department of Physics and Astronomy, Seoul National University, Seoul 08826, Korea}
\author{Kabgyun Jeong}
\email{kgjeong6@snu.ac.kr}
\affiliation{Research Institute of Mathematics, Seoul National University, Seoul 08826, Korea}
\affiliation{School of Computational Sciences, Korea Institute for Advanced Study, Seoul 02455, Korea}

\date{\today}

\begin{abstract}
We study the $p$-R\'{e}nyi entropy power inequality with a weight factor $t$ on two independent continuous random variables $X$ and $Y$. The extension essentially relies on a modulation on the sharp Young's inequality due to Bobkov and Marsiglietti. Our research provides a key result that can be used as a fundamental research finding in quantum Shannon theory, as it offers a R\'{e}nyi version of the entropy power inequality for quantum systems.
\end{abstract}

\maketitle

\section{Introduction} \label{intro}
As a relevant measure of information content, Shannon introduced the differential entropy~\cite{S48} in the form of
\begin{equation}
	H(X)=-\int_{\bb{R}^d}p_X(x)\log p_X(x)\tn{d}x,
\end{equation}
where $X\in\bb{R}^d$ is a continuous random variable with a probability density function $p_X$. It reveals that this quantity is inherently the same as thermodynamic entropy, as proposed by Boltzmann. Shannon's novel perspective seamlessly merges two distinct realms: information theory and physics. These disciplines have been instrumental in driving numerous groundbreaking discoveries, including Bell's inequality~\cite{B64} and the area law~\cite{ECP10}. The significance of entropy persists to this day, exemplified by Von Neumann entropy, which measures quantum entanglement~\cite{CC05}.

Furthermore, Shannon proposed another quantity, which is called the associated entropy power
\begin{equation}
	\mathbf{V}(X)=\exp\left(\frac{2}{d}H(X)\right).
\end{equation}
These entropic quantities possess several fundamental properties, with one of the most notable being the entropy power inequality (EPI):
\begin{equation} \label{eq:epi}
	\mathbf{V}(X+Y)\ge\mathbf{V}(X)+\mathbf{V}(Y).
\end{equation}
The first complete proof for the inequality~\eqref{eq:epi} was provided by Stam~\cite{S59}, and detailed proofs by Blachman~\cite{B65}. The second approach for the proof was given by Beckner~\cite{W75}, and Brascamp and Lieb~\cite{BL76,L78}, and the final one was recently obtained by Rioul~\cite{R17}. We can classify and review those three-type of the mathematical proofs on the EPI, including several variants, as follows: That is, it is given by
\begin{itemize}
	\item[$\bullet$] Gaussian perturbation, de Bruijn's identity, and convexity of Fisher information~\cite{S59,B65,CS91,DSV06,VG06,B07,R07,R11}; 
	\item[$\bullet$] Sharp Young's convolution inequality~\cite{W75,BL76,L78,SV00,WM13,WM14}; and
	\item[$\bullet$] Change of variables on convex bodies~\cite{R17}.  
\end{itemize}

It was also known that above EPI~\eqref{eq:epi} is essentially equivalent to the linear form of
\begin{equation} \label{eq:lepi}
	H(\sqrt{t}X+\sqrt{1-t}Y)\ge tH(X)+(1-t)H(Y)
\end{equation}
for any weight-parameter $t\in(0,1)$, and we can find the proofs of the equivalence relation in Refs.~\cite{VG06,R11,DCT91,MMX17}.
\bigskip
\begin{remark} \label{rmk1}
	For any random variable $X\in\bb{R}^d$ and for any $t>0$, the differential entropy and its entropy power satisfy the scaling properties
	\begin{equation*}
		H(\sqrt{t}X)=H(X)+\frac{d}{2}\log t\;\;\; \tn{and}\;\;\; \mathbf{V}(\sqrt{t}X)=t\mathbf{V}(X).
	\end{equation*}
\end{remark}

There are various entropic quantities, including min entropy, max entropy, collision entropy, and more. In 1961, Alfr\'{e}d R\'{e}nyi introduced R\'{e}nyi entropy, a concept that generalizes these entropies while preserving the additivity of independent variables~\cite{R61}.
Now, we delve into R\'{e}nyi (differential) entropy and the associated entropy power as an extension of Shannon's concepts. These are defined as follows, for any order $p>0$,
\begin{align}
	H_p(X) &= \frac{1}{1-p}\log\int_{\mathbb{R}^d}p_X^{p}(x)\mathrm{d}x, \\
	\mathbf{V}_p(X) &= \exp\left(\frac{2}{d}H_p(X)\right).
\end{align}
R\'{e}nyi entropy becomes Shannon entropy when the order $p$ approaches 1. Beyond its versatility in expressing various forms of entropy, R\'{e}nyi entropy also establishes a direct connection with free energy in statistical mechanics, where the order $p$ is related to temperature~\cite{B22,FG22}.
Let us consider a probability density function for the thermal equilibrium state at temperature $T_0$, given by
\begin{equation}
	p_0(r, p) = e^{-\mathcal{H}(r, p)/kT_0},
\end{equation}
Here, $\mathcal{H}(r, p)$ represents the Hamiltonian of this state in phase space, and $T_0$ is chosen such that it satisfies the partition function $Z(T_0) = 1$.
Now, if we have the system's Gibbs state defined as
\begin{equation}
	p(r, p) = \frac{e^{-\mathcal{H}(r, p)/kT}}{Z(T)},
\end{equation}
and we express the free energy under this assumption as $F(T)=-kT\ln Z(T)$, we can derive the R\'{e}nyi entropy of order $T_0/T$ as follows:
\begin{equation}
	H_{T_0/T} = -\frac{F(T)}{T-T_0}.
\end{equation}
In the discrete case, R\'{e}nyi entropy is utilized for measuring entanglement entropy~\cite{IMPPTMAM15} and featuring in conformal field theory (CFT)~\cite{P14}. Consequently, R\'{e}nyi entropy is poised to serve as another critical link between information theory and physics.

For these entropic functionals, it was known that~\cite{BM17}
\begin{equation} \label{eq:repi}
	\mathbf{V}_p^\alpha(X+Y)\ge\mathbf{V}_p^\alpha(X)+\mathbf{V}_p^\alpha(Y),
\end{equation}
where $\alpha\ge\frac{p+1}{2}$, and it is called $p$-R\'{e}nyi entropy power inequality of power $\alpha$. However, this inequality~\eqref{eq:repi} has a restriction on the order of $p>1$ --- further it is not true when $p=\infty$. There are several variants of the R\'{e}nyi entropy power inequality~\cite{BC15,RI16}, and the scaling properties~\cite{ST14} are as follows.

\bigskip
\begin{remark} \label{rmk2}
	For any random variable $X\in\mathbb{R}^d$ and $t>0$, the following scaling properties hold for the R\'{e}nyi entropy and associated entropy power:
	\begin{equation} \label{eq:scale}
		H_p(\sqrt{t}X)=H_p(X)+\frac{d}{2}\log t\;\;\; \tn{and}\;\;\; \mathbf{V}_p(\sqrt{t}X)=t^\mu\mathbf{V}_p(X),
	\end{equation}
	where $\mu := (p-1)d+2$ is a positive coefficient.
\end{remark}
\bigskip

We notice that, while Savar\'{e} and Toscani choose the power as $\alpha\ge1+(p-1)\frac{n}{2}$~\cite{ST14}, Bobkov and Marsiglietti take as $\alpha\ge\frac{p+1}{2}$~\cite{BM17} with $p>1$. For convenience, we fix $\mu$ to match $\alpha$ through the context in Ref.~\cite{BM17}.

The aim of this paper is to give a simple extension of the $p$-R\'{e}nyi entropy power inequality, derived by Bobkov and Marsiglietti~\cite{BM17}, in the forms with a weight factor $t$. Before the main proof, we briefly make a summary of the research trend on quantum entropy power inequalities on which we become aware of why the EPIs are essential and powerful in (quantum) information theory. The EPI is a critical concept in quantum Shannon theory, as it plays a significant role in understanding the fundamental limits of quantum communication by providing a framework for analyzing channel capacities as well as detecting a (potential) entanglement in a continuous-variable (CV) quantum system. Our research is highly significant in the field of quantum information theory because it offers a new interpretation of the EPI in quantum systems. Our findings have the potential to contribute significantly to the field by developing a novel quantum entropy power inequality.

Before moving on to the main content, we define the convolution-related symbols used in this paper. The general principle states that for any independent random variables $X$ and $Y$ in $\mathbb{R}^d$, where a convolution operation is denoted as $\boxplus$, it maps $(p_X, p_Y)$ to $p_{X\boxplus Y} := p_Z$ such that $p_Z(z) = \int_{\mathbb{R}^d} p_X(x) p_Y(z-x) \mathrm{d}x$. Moreover, the weighted sum of random variables is defined as $X\boxplus_t Y := \sqrt{t}X + \sqrt{1-t}Y$.
Regarding the quantum convolution operation denoted as $\boxplus_\tau$, it is defined by mapping $(\rho_X, \rho_Y)$ to $\rho_{X\boxplus_\tau Y} := \rho_Z$ in a manner such that $\rho_X^G\boxplus_\tau\rho_Y^G = \T_Y[U_\tau^G(\rho_X^G\otimes\rho_Y^G)U_\tau^{G^\dag}]$, where $U_\tau$ represents any unitary operator acting on the composite quantum system $XY$.

\section{Entropy Power Inequalities in Quantum Regime}\label{sec:qepi}
We briefly introduce entropy power inequalities in quantum Shannon theory. A quantum version of the entropy power inequality is first derived by K\"{o}nig and Smith~\cite{KS14}, and extended to general case~\cite{PMG14} in which they use two independent continuous random variables on bosonic Gaussian systems. While the conventional entropy power inequalities make use of the usual convolution operation, most quantum cases use a beam-splitting or amplifying operation to mix random variables or quantum states. Recently, there are few studies on the case of conditional entropy power inequality on quantum systems~\cite{K15,JLJ18,PT18,P19} as well as a special case for $d$-dimensional quantum bit (namely, qudit)~\cite{ADO16} for highlighting entropy photon number inequality~\cite{GES08}---yet it is unproved.

In information theory, while the channel capacity is always additive, it was reported that quantum channel capacities on quantum channels are \emph{non-additive} via essential quantum effects such as quantum entanglement or super-activation~\cite{SY08,H09,LWZG09,LLSSS23}. This means that determining the quantum channel capacities are extremely hard problem~\cite{H06}. However, it was known that quantum entropy power inequalities can help to obtain a tight upper bound on the quantum channel capacities~\cite{KS13,KS13+,HK18,JLL19,LLKJ19,J20}. We believe that an intimate collaboration between conventional information and quantum information societies could boost the deep understanding for a fundamental and unsolved information-theoretic problems. 

Now, let us return to our main proof for the weighted version of the R\'{e}nyi entropy power inequality, and propose an important conjecture on the entropy power inequality. The subsequent proof plays a crucial role in presenting a new version of the R\'{e}nyi entropy power inequality for bosonic Gaussian quantum systems, corresponding to Theorem \ref{thm:QREPIeq} in the recent paper~\cite{QREPI}.

\bigskip
\begin{theorem}[Quantum $p$-R\'{e}nyi EPI~\cite{QREPI}] \label{thm:QREPIeq}
	For any Gaussian states $\rho_X^G$ and $\rho_Y^G$ in the $D$-mode space $\mathrm{Sp}(2D,\mathbb{R})$, and for any $\tau\in(0,1)$, the following inequality holds in the context of quantum information theory:
	\begin{equation} \label{eq:qrepi-p}
		\mathbf{V}_p^\kappa(\rho_X^G\boxplus_\tau\rho_Y^G)\ge \tau^\kappa\mathbf{V}_p^\kappa(\rho_X^G)+(1-\tau)^\kappa\mathbf{V}_p^\kappa(\rho_Y^G),
	\end{equation}
	where $\boxplus_\tau$ denotes the quantum convolution operation, and $\kappa \ge \frac{p+1}{2}~(\forall p>1)$.
\end{theorem}

\section{R\'{e}nyi Entropy Power Inequality with Weight}
Let $X$ and $Y$ in $\bb{R}^d$ be independent random variables with probability density functions $p_X$ and $p_Y$, respectively. Let us all parameter $p,q,r\ge1$, and those are satisfying the condition $\frac{1}{q}+\frac{1}{r}-\frac{1}{p}=1$. Then, the famous Young's inequality is given by~\cite{W75,BL76}
\begin{equation*}
	\|p_X* p_Y\|_p\le C^{\frac{d}{2}}(p,q,r)\|p_X\|_q\|p_Y\|_r,
\end{equation*}
where $*$ denotes the convolution on densities, $\|\cdot\|_s$ the $L^s$-norm on the non-negative function on $\bb{R}^d$, and thus, $\|p_X\|_s=\mathbf{V}_s(X)^{-\frac{d}{2}\left(1-\frac{1}{s}\right)}$. For convenience, we denote $C=C(p,q,r)$. From the Young's inequality above, we have

\begin{equation} \label{eq:young1}
	\mathbf{V}_p^{1-\frac{1}{p}}(X\boxplus Y)\ge\frac{1}{C}\mathbf{V}_q^{1-\frac{1}{q}}(X)\mathbf{V}_r^{1-\frac{1}{r}}(Y),
\end{equation}
where $X\boxplus Y$ defines the usual convolution operation on random variables $X$ and $Y$. By exploiting mathematical techniques such as H\"{o}lder's inequality and optimization, we can change Eq.~\eqref{eq:young1} into

\begin{equation} \label{eq:young2}
	\mathbf{V}_p^\alpha(X\boxplus Y)\ge{C^{-\frac{\alpha p}{p-1}}}\mathbf{V}_p^{\frac{\alpha p(q-1)}{q(p-1)}}(X)\mathbf{V}_p^{\frac{\alpha p(r-1)}{r(p-1)}}(Y).
\end{equation}
The proof of R\'{e}nyi entropy power inequality can be completed via the following key lemma (Lemma~\ref{lem:const}) under $a:=\mathbf{V}_p^\alpha(X)$ and $b:=\mathbf{V}_p^\alpha(Y)$. (See proof details in~\cite{BM17}.)

\bigskip
\begin{lemma}[Sharp Young's constant~\cite{BM17}]\label{lem:const}
	Let $a,b>0$ satisfy $a+b=1-\frac{1}{p}$. For any $p>1$, there exists $q,r\ge1$ satisfying $\frac{1}{q}+\frac{1}{r}-\frac{1}{p}=1$ such that
	\begin{align}
		{C^{-\frac{\alpha p}{p-1}}}a^{\frac{p(q-1)}{q(p-1)}}b^{\frac{p(r-1)}{r(p-1)}}\ge1-\frac{1}{p}, \label{eq:modyoung}
	\end{align}
	where $C$ is an optimal Young's constant, and $\alpha=\frac{p+1}{2}$.
\end{lemma}

This directly implies that the $p$-R\'{e}nyi entropy power inequality with power $p$ holds as follows.

\bigskip
\begin{theorem}[Bobkov and Marsiglietti~\cite{BM17}]
	For any independent random variables $X,Y\in\bb{R}^d$ with probability density functins $p_X,p_Y$ respectively, we have
	\begin{equation}
		\mathbf{V}_p^\alpha(X\boxplus Y)\ge\mathbf{V}_p^\alpha(X)+\mathbf{V}_p^\alpha(Y),
	\end{equation}
	where $\boxplus$ denotes the convolution operation, and $\alpha\ge\frac{p+1}{2}~(\forall p>1)$.
\end{theorem}

Now, we take into account Eq.~\eqref{eq:young2} including a weight $t\in(0,1)$, that is,
\begin{align} \label{eq:young3}
	&\mathbf{V}_p^\alpha(X\boxplus_t Y) =\mathbf{V}_p^\alpha(\sqrt{t}X+\sqrt{1-t}Y) \\
	&\ge{C^{-\frac{\alpha p}{p-1}}}\mathbf{V}_p^{\frac{\alpha p(q-1)}{q(p-1)}}(\sqrt{t}X)\mathbf{V}_p^{\frac{\alpha p(r-1)}{r(p-1)}}(\sqrt{1-t}Y) \\
	&={C^{-\frac{\alpha p}{p-1}}}\{t^\mu\mathbf{V}_p(X)\}^{\frac{\alpha p(q-1)}{q(p-1)}}\{(1-t)^\mu\mathbf{V}_p(Y)\}^{\frac{\alpha p(r-1)}{r(p-1)}} \\
	&={C^{-\frac{\alpha p}{p-1}}}t^{\frac{\mu\alpha p(q-1)}{q(p-1)}}\{\mathbf{V}_p^a(X)\}^{\frac{p(q-1)}{q(p-1)}}(1-t)^{\frac{\mu\alpha p(r-1)}{r(p-1)}}\{\mathbf{V}_p^a(Y)\}^{\frac{p(r-1)}{r(p-1)}} \\
	&:={C^{-\frac{\alpha p}{p-1}}}t^{\beta_1}a^{\frac{p(q-1)}{q(p-1)}}(1-t)^{\beta_2}b^{\frac{p(r-1)}{r(p-1)}},  \label{eq:young5}
\end{align}
where $X\boxplus_t Y$ denotes the weight convolution on the random variables $X$ and $Y$. The equality in Eq.~\eqref{eq:young3} is directly given by the scaling properties in Eq.~\eqref{eq:scale}: That is, for any $s>1$ and for some $\mu$ (Remark~\ref{rmk2}),
\begin{equation*}
	\left\|p_{\sqrt{t}X}\right\|_s:=\mathbf{V}_s^{-\frac{d}{2}\left(1-\frac{1}{s}\right)}(\sqrt{t}X)=t^{-\frac{\mu d}{2}\left(1-\frac{1}{s}\right)}\mathbf{V}_s^{-\frac{d}{2}\left(1-\frac{1}{s}\right)}(X).
\end{equation*} 
If we define $\tilde{a}:=t^\alpha\mathbf{V}_p^\alpha(X)=t^\alpha a$ and $\tilde{b}:=(1-t)^\alpha\mathbf{V}_p^\alpha(Y)=(1-t)^\alpha b$ under a proper choice of $\beta_1$ and $\beta_2$ such that $\alpha\ge\frac{p+1}{2}$ for any $t\in(0,1)$, respectively, then we have

\bigskip
\begin{prop} \label{prop:main}
	Let $\tilde{a},\tilde{b}>0$ satisfy $\tilde{a}+\tilde{b}=1-\frac{1}{p}$. For any $p>1$ and for any $t\in(0,1)$, there exists $q,r\ge1$ satisfying $\frac{1}{q}+\frac{1}{r}-\frac{1}{p}=1$ such that
	\begin{align}
		{C^{-\frac{\alpha p}{p-1}}}\tilde{a}^{\frac{p(q-1)}{q(p-1)}}\tilde{b}^{\frac{p(r-1)}{r(p-1)}}\ge1-\frac{1}{p}, \label{eq:modyoung}
	\end{align}
	where $C=C(p,q,r)$ is an optimal Young's constant, and $\alpha=\frac{p+1}{2}$.
\end{prop}

\begin{proof}
	The proof is essentially equivalent to the proof of Lemma~\ref{lem:const}. See details of the proof in Ref.~\cite{BM17}.
\end{proof}

The proof of Theorem~\ref{thm:J} (below) is straightforward from Proposition~\ref{prop:main}, as an extension of the sharp Young's constant in Lemma~\ref{lem:const}.

\bigskip
\begin{theorem}[Weighted $p$-R\'{e}nyi Entropy Power Inequality] \label{thm:J}
	Let $X$ and $Y$ in $\bb{R}^d$ be any independent random variables with probability density functions $p_X$ and $p_Y$, respectively. For any $t\in(0,1)$, then we have
	\begin{equation} \label{eq:wrepi}
		\mathbf{V}_p^\alpha(X\boxplus_t Y)\ge t^\alpha\mathbf{V}_p^\alpha(X)+(1-t)^\alpha\mathbf{V}_p^\alpha(Y),
	\end{equation}
	where $\boxplus_t$ denotes the weighted sum of random variables, and $\alpha\ge\frac{p+1}{2}~(\forall p>1)$.
\end{theorem}

More precisely, the inequality in Eq.~\eqref{eq:wrepi} has in the form of
\begin{align*}
	\exp\left(\frac{2\alpha}{d}H_p(X\boxplus_t Y)\right)&\ge t^\alpha\exp\left(\frac{2\alpha}{d}H_p(X)\right)\\
	&+(1-t)^\alpha\exp\left(\frac{2\alpha}{d}H_p(Y)\right),
\end{align*}
where $H_p(\cdot)$ denotes the $p$-R\'{e}nyi entropy.



\section{Conclusions} \label{conclusion}
In the paper, we provided a proof of the R\'{e}nyi entropy power inequality with an appropriate weight factor. In a sense, this inequality can be viewed as a corollary of the main result in Bobkov and Marsiglietti's theorem, although there is not an exact statement to that effect.

Since free energy is connected to R\'{e}nyi entropy, our proposed inequality has the potential to inspire new relationships between thermodynamic quantities. Thanks to its consistency, a quantum analogue is also feasible, opening up new vistas for understanding the intersection of information theory and physics.

We still have several open questions regarding entropy power inequalities in both quantum and classical regimes. Here are some of them:

\begin{itemize}
	\item Can we extend the Entropy Power Inequalities (EPIs) beyond the R\'{e}nyi case to include other entropy measures, such as Tsallis entropy?
	\item Is it possible to prove the (weighted) $p$-R\'{e}nyi EPI using alternative mathematical techniques, without relying on Young's inequality?
	\item How can we construct a quantum convolution, analogous to the classical convolution, to establish new types of quantum entropy power inequalities?
	\item What is the true capacity of a given quantum channel?~\cite{SY08,H09,LWZG09,LLSSS23} or can we efficiently detect a (non-)Gaussian entanglement?~\cite{CMYY23}
\end{itemize}

Answering these questions requires close collaboration between the classical and quantum information communities, with a focus on studying EPIs. This research is essential not only for advancing pure mathematical theory, including the development of more precise bounds and related analytical fields, but also because it directly relates to the optimal design and practical implementation of quantum channels through mathematical analysis. Therefore, this topic represents one of the most significant challenges~\cite{BGJ23} in quantum Shannon theory.

\bigskip

\section*{Acknowledgments}
This work was supported by the National Research Foundation of Korea (NRF) through a grant funded by the Ministry of Science and ICT (NRF-2022M3H3A1098237 \& RS-2023-00211817) and the Ministry of Education (NRF-2021R1I1A1A01042199). This work was partly supported by Institute for Information \& communications Technology Promotion (IITP) grant funded by the Ministry of Science and ICT (No. 2019-0-00003), and Korea Institute of Science and Technology Information (KISTI).

%


\end{document}